\documentclass[12pt]{article}%
\usepackage{prefix-private}

\bibliography{max_density}%

\title{Approximating Densest Subgraph in Geometric Intersection
   Graphs}

\author{%
   Sariel Har-Peled%
   \thanks{Department of Computer Science; University of Illinois; 201
      N. Goodwin Avenue; Urbana, IL, 61801, USA;
      \href{mailto:spam@illinois.edu}{sariel@illinois.edu};
      \url{http://sarielhp.org/}.}%
   \and%
   Saladi Rahul%
   \thanks{Indian Institute of Science (IISc);
      \url{saladi@iisc.ac.in}. }%
}

   \date{\today}%

\begin{document}

\maketitle

\begin{abstract}
    For an undirected graph $\mathsf{G}=(\mathsf{V}, \mathsf{E})$, with $n$ vertices and $m$
    edges, the \emph{densest subgraph} problem, is to compute a subset
    $S \subseteq \mathsf{V}$ which maximizes the ratio
    ${\left| {\mathsf{E}_S} \right|} / {\left| {S} \right|}$, where $\mathsf{E}_S \subseteq \mathsf{E}$ is
    the set of all edges of $\mathsf{G}$ with endpoints in $S$.  The densest
    subgraph problem is a well studied problem in computer science.
    Existing exact and approximation algorithms for computing the
    densest subgraph require $\Omega(m)$ time.  We present near-linear
    time (in $n$) approximation algorithms for the densest subgraph
    problem on \emph{implicit} geometric intersection graphs, where
    the vertices are explicitly given but not the edges.  As a
    concrete example, we consider $n$ disks in the plane with
    arbitrary radii and present two different approximation
    algorithms.
\end{abstract}

\section{Introduction}

Given an undirected graph $\mathsf{G}=(\mathsf{V},\mathsf{E})$, the \emphi{density} of a set
$S \subseteq V(\mathsf{G})$ is ${\left| {\mathsf{E}_S} \right|}/{\left| {S} \right|}$, where each
edge in $\mathsf{E}_S \subseteq \mathsf{E}\mleft({\mathsf{G}}\mright)$ has both its vertices in $S$. In
the \emphw{densest subgraph problem}, the goal is to find a subset of
$V(\mathsf{G})$ with the maximum density. Computing the densest subgraph is a
primitive operation in large-scale graph processing, and has found
applications in mining closely-knit communities~\cite{cs-dseac-12},
link-spam detection~\cite{gkt-dldsm-05}, and reachability and distance
queries~\cite{chkz-rdq2l-03}.  See~\cite{bkv-dssm-12} for a detailed
discussion on the applications of densest subgraph, and
\cite{lmfb-sdsp-23} for a survey on the recent developments on densest
subgraph.

\paragraph*{Exact algorithms for densest subgraph.}

Unlike the (related) problem of computing the largest clique (which is
\NPHard), the densest subgraph can be computed (exactly) in polynomial
time.  Goldberg~\cite{g-fmds-84} show how to reduce the problem to
$O(\log n)$ instances of $s{-}t$ min-cut problem. Gallo \textit{et~al.}\xspace
\cite{ggt-fpmfa-89} improved the running time slightly by using
parametric max flow computation.  Charikar~\cite{c-gaafd-00} presented
an LP based solution to solve the problem, for which Khuller and
Saha~\cite{ks-fds-09} gave a simpler rounding scheme.

\paragraph*{Approximation algorithms for densest subgraph.} The exact
algorithms described above require solving either an LP or an $s{-}t$
min-cut instance, both of which are relatively expensive to compute.
To obtain a faster algorithm, Charikar~\cite{c-gaafd-00} analyzed a
$2$-approximation algorithm which repeatedly removes the vertex with
the smallest degree and calculates the density of the remaining
graph. This algorithm runs in linear time. After that, Bahmani \textit{et~al.}\xspace
\cite{bgm-epgam-14} used the primal-dual framework to give a
$(1+\varepsilon)$-approximation algorithm which runs in
$O\mleft({ (m/\varepsilon^2)\log n}\mright)$ time.  The problem was also studied in the
streaming model the focus is on approximation using bounded space.
McGregor \textit{et~al.}\xspace \cite{mtvv-dsdgs-15} and Esfandiari \textit{et~al.}\xspace
\cite{ehw-baaus-16} presented a $(1+\varepsilon)$-approximation streaming
algorithm using roughly $\widetilde{O}(\left| {\mathsf{V}} \right|)$ space.  There
has been recent interest in designing {\em dynamic} algorithms for
approximate densest subgraph, where an edge gets inserted or deleted
in each time step \cite{bhnt-stamd-15,sw-nfdds-20, els-edsc-15}.

\paragraph*{Geometric intersection graphs.} The geometric intersection
graph of a set $\mathcal{D}$ of $n$ objects is a graph $G=(V,E)$, where each
object in $\mathcal{D}$ corresponds to a unique vertex in $V$, and an edge
exists between two vertices if and only if the corresponding objects
intersect.  Further, in an {\em implicit geometric intersection
   graph}, the input is only the set of objects $\mathcal{D}$ and {\em not}
the edge set $\mathsf{E}$, whose size could potentially be quadratic in terms
of $\left| {\mathsf{V}} \right|$.

Unlike general graphs, the geometric intersection graphs typically
have more structure, and as such computing faster approximation
algorithms~\cite{ap-naghs-14,chq-flbag-20}, or obtaining better
approximation algorithms on implicit geometric intersection graphs has
been an active field of research.

One problem closely related to the densest subgraph problem is that of
finding the maximum clique.  Unlike the densest subgraph problem, the
maximum clique problem is NP-hard for various geometric intersection
graphs as well (for e.g., segment intersection
graphs~\cite{ccl-cprig-13}).  However, for the case of unit-disk
graphs, an elegant polynomial time solution by Clark \textit{et~al.}\xspace
\cite{ccj-udg-90} is known.

\begin{figure}[h]
    \centering%
    \includegraphics{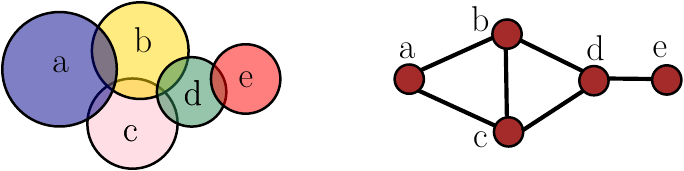}
    \caption{Five disks and their corresponding graph. The densest
       subgraph is $\{a,b,c,d\}$ with density $5/4$.}
    \label{fig:gig}
\end{figure}

\paragraph{Motivation.}
Densest subgraph computation on geometric intersection graphs can help
detect regions which have strong cellular network coverage, or have
many polluting factories.  The region covered (resp., polluted) by
cellular towers (resp., or factories) can be represented by disks of
varying radii.

\subsection{Problem statement and our results}

In this paper we study the densest subgraph problem on implicit
geometric intersection graphs, and present near-linear time (in terms
of $\left| {\mathsf{V}} \right|$) approximation algorithms.

\paragraph{From reporting to approximate counting/sampling.}

We show a reduction from (shallow) range-reporting to approximate
counting/sampling. Previous work on closely related problem includes
the work by Afshani and Chan \cite{ac-ohrrt-09} (that uses shallow
counting queries), and the work by Afshani \textit{et~al.}\xspace
\cite{ahz-gacor-10}. The reduction seems to be new, and should be
useful for other problems.  See \hyperref[sec:app-sampling]{Section~\ref*{sec:app-sampling}} and
\hyperref[theo:gen-red-I]{Theorem~\ref*{theo:gen-red-I}} for details.

Importantly, this data-structure enables us to sample
$(1\pm \varepsilon)$-uniformly a disk from the set of disks intersecting a
given disk.

\paragraph{The application.}

For the sake of concreteness, we consider the case of $n$ disks (with
arbitrary radii) lying in the plane and present two different
approximation algorithms.  See \hyperref[fig:gig]{Figure~\ref*{fig:gig}}.

\paragraph{A $(2+\varepsilon)$-approximation.}

Our first approximation algorithm uses the greedy strategy of removing
disks of low-degree from the intersection graph. By batching the
queries, and using the above data-structure, we get a
$(2+\varepsilon)$-approximation for the densest subset of disks in time
$O_\varepsilon( n \log^3 n)$, where $O_\varepsilon$ hides constants polynomial in
$1/\varepsilon$.

\paragraph{A $(1+\varepsilon)$-approximation.}

A more promising approach is to randomly sample edges from the
intersection graph, and then apply known approximation
algorithms. This requires some additional work since unlike previous
work, we can only sample approximately in uniform, see
\hyperref[sec:expensive]{Section~\ref*{sec:expensive}} for details.  The running time of the new algorithm
is $O_\varepsilon( n\log^2 n \log \log n)$ which is faster than the first
inferior $(2+\varepsilon)$-approximation. The results are summarized in
\hyperref[fig:results]{Figure~\ref*{fig:results}}.

\begin{figure}
    \centering%
    \begin{tabular}{c|l|c}
      \hline
      Approximation
      & Running time
      & Ref
      \\%
      \hline%
      $2+\varepsilon$
      &%
        \begin{math}
            O\bigl(%
            \frac{n\log^2n}{\varepsilon^2}%
            (  \tfrac{1}{\varepsilon^2} + \log n)%
            \bigr)\Bigr.
        \end{math}
      &
        \hyperref[theo:first]{Theorem~\ref*{theo:first}}
      \\%
      \hline%
      $1+\varepsilon$
      &
        \begin{math}
            O\bigl(%
            \frac{n\log^2n}{\varepsilon^2} ( \tfrac{1}{\varepsilon^2} + \log\log n)
            \bigr)\Bigr.
        \end{math}
      &
        \hyperref[theo:second]{Theorem~\ref*{theo:second}}
      \\
      \hline%
    \end{tabular}
    \caption{Our results.}
    \label{fig:results}
\end{figure}

\section{Preliminaries}

\subsection{Definitions}

In the following, $\mathcal{D}$ denotes a (given) set of $n$ objects (i.e.,
disks).  For $S \subseteq \mathcal{D}$ and $u \in \mathcal{D}$, we use the
shorthands $S + u = S \cup \{ u\}$ and $S - u = S \setminus \{ u\}$.
For $\varepsilon \in (0,1)$ and a real number $\alpha > 0$, let
$(1\pm\varepsilon)\alpha$ denote the interval
$\bigl((1-\varepsilon) \alpha, (1+\varepsilon)\alpha\bigr)$.  Throughout, a
statement holds \emphi{with high probability}, if it holds with
probability at least $1-n^{-c}$, where $c$ is a sufficiently large
constant.

\begin{observation}
    \label{observation:eps:games}%
    \begin{compactenumI}
        \smallskip%
        \item For any $\varepsilon$, we have $\frac{1}{1+\varepsilon}\geq 1-\varepsilon$.
        \smallskip
        \item For $\varepsilon \in (0,1/2)$, we have
        \begin{math}
            \frac{1}{1\pm \varepsilon} =%
            \bigl( 1/(1+\varepsilon), 1/(1-\varepsilon)\bigr)%
            \subseteq%
            1\pm 2\varepsilon
        \end{math}
        since
        \begin{math}
            1-\varepsilon \leq \frac{1}{1+\varepsilon} \leq \frac{1}{1-\varepsilon} \leq
            1+2\varepsilon.
        \end{math}
        \smallskip%
        \item For $\varepsilon \in (0,1/3)$, we have
        $(1\pm \varepsilon)^2 \subseteq (1\pm 3\varepsilon)$.  \smallskip

        \item For $\varepsilon \in (0,1)$ and constants $c, c_1$ and $c_2$,
        such that $c \geq c_1c_2$, we have
        \begin{math}
            (1\pm c_1\varepsilon/c)(1\pm c_2\varepsilon/c)%
            \subseteq%
            1\pm \frac{c_1+c_2+1}{c}\varepsilon
        \end{math}%
        \footnote{Indeed,
           \begin{math}
               (1 + c_1\varepsilon/c)(1 + c_2\varepsilon/c)%
               \leq%
               1 + (c_1/c + c_2/c + c_1c_2/c^2)\varepsilon \leq%
               1 + (c_1 + c_2 + 1)\varepsilon/c.
           \end{math}%
        }.
    \end{compactenumI}
\end{observation}

\begin{defn}
    Given a set of objects $\mathcal{D}$ (say in $\mathbb{R}^d$), their
    \emphi{intersection graph} has an edge between two objects if and
    only if they intersect. Formally,
    \begin{equation*}
        \mathsf{G}_{\cap \mathcal{D}}
        =%
        \mleft({\mathcal{D},
           \left\{  u v \;\middle\vert\; u,v \in \mathcal{D}, \text{ and } u \cap v
              \neq \emptyset  \right\} }\mright).
    \end{equation*}
\end{defn}

\begin{defn}
    For a graph $\mathsf{G}$, and a subset $S \subseteq \mathsf{V}\mleft({\mathsf{G}}\mright)$, let
    \begin{equation*}
        \mathsf{E}_{S}
        =
        \mathsf{E}_{S}\mleft({\mathsf{G}}\mright)
        =
        \left\{  uv \in \mathsf{E}\mleft({\mathsf{G}}\mright) \;\middle\vert\; u,v \in S \right\}.
    \end{equation*}
    The \emphi{induced subgraph} of $\mathsf{G}$ over $S$ is
    $\mathsf{G}_{S} = (S, \mathsf{E}_{S})$, and let $\mathsf{m}\mleft({S}\mright) = \left| {\mathsf{E}_{S}} \right|$
    denote the number of edges in this subgraph.

\end{defn}

\begin{defn}
    For a set $S \subseteq \mathsf{V}\mleft({\mathsf{G}}\mright)$, its \emphi{density} in $\mathsf{G}$ is
    \begin{math}
        {\nabla}\mleft({S}\mright) =%
        \nabla_{\!\!\mathsf{G}}\mleft({S}\mright) =%
        \mathsf{m}\mleft({\mathsf{G}_{S}}\mright) / \left| { S} \right|,
    \end{math}
    where $\mathsf{m}\mleft({\mathsf{G}_{S}}\mright)$ is the number of edges in $\mathsf{G}_{S}$.
    Similarly, for a set of objects $\mathcal{D}$, and a subset
    $S \subseteq \mathcal{D}$, the \emphi{density} of $S$ is
    ${\nabla}\mleft({S}\mright) = \mathsf{m}\mleft({\mathsf{G}_{\cap S}}\mright) / \left| { S} \right|$.
\end{defn}

\begin{defn}
    For a graph $\mathsf{G}$, its \emphi{max density} is the quantity
    \begin{math}
        \mathcalb{d}\mleft({\mathsf{G}}\mright)
        = \max_{S \subseteq\mathsf{V}\mleft({\mathsf{G}}\mright)} {\nabla}\mleft({S}\mright),
    \end{math}
    and analogously, for a set of objects $\mathcal{D}$, its \emphi{max
       density} is
    \begin{math}
        \mathcalb{d}\mleft({\mathcal{D}}\mright)
        = \max_{S \subseteq\mathcal{D}} {\nabla}\mleft({S}\mright),
    \end{math}
\end{defn}

The problem at hand is to compute (or approximate) the maximum density
of a set of objects $\mathcal{D}$.  If a subset $S$ realizes this quantity,
then it is the \emphi{densest subset} of $\mathcal{D}$ (i.e.,
$\mleft({\mathsf{G}_{\cap \mathcal{D}}}\mright)_{S}$ is the densest subgraph of
$\mathsf{G}_{\cap \mathcal{D}}$). One can make the densest subset unique, if there are
several candidates, by asking for the lexicographic minimal set
realizing the maximum density. For simplicity of exposition we threat
the densest subset as being unique.

\begin{lemma}
    \label{lemma:density:l:b}%
    Let $\mathcal{O} \subseteq \mathcal{D}$ be the densest subset, and let
    $\nabla = {\nabla}\mleft({\mathcal{O}}\mright)$. Then, for any object $u \in \mathcal{O}$, we
    have $d_{\mathcal{O}}\mleft({u}\mright) = \left| { u \sqcap (\mathcal{O}-u)} \right| \geq \nabla$,
    where
    $u \sqcap (\mathcal{O}-u) = \left\{  x \in \mathcal{O} - u \;\middle\vert\; x \cap u \neq
       \emptyset \right\}$.
\end{lemma}
\begin{proof}
    Observe that
    \begin{equation*}
        \nabla
        =
        \frac{\left| {\mathsf{E}_{\mathcal{O}}} \right|}{\left| {\mathcal{O}} \right|}
        =
        \frac{\left| {\mathsf{E}_{\mathcal{O}-u}} \right| + d_{\mathcal{O}}\mleft({u}\mright)}{\left| {\mathcal{O}} \right| -1 + 1}.
    \end{equation*}
    As such, if $d_{\mathcal{O}}\mleft({u}\mright) < \nabla$, then
    \begin{equation*}
        \frac{d_{\mathcal{O}}\mleft({u}\mright)}{1}
        <
        \nabla
        =
        \frac{d_{\mathcal{O}}\mleft({u}\mright) + \left| {\mathsf{E}_{\mathcal{O}-u}} \right|}{1 + \left| {\mathcal{O}} \right| -1 }
        <
        \frac{\left| {\mathsf{E}_{\mathcal{O}-u}} \right|}{\left| {\mathcal{O}} \right| -1 }.
    \end{equation*}
    But this implies that $\mathcal{O} - u$ is denser than $\mathcal{O}$, which is a
    contradiction.
\end{proof}

\subsection{Reporting all intersecting pairs of disks}

The algorithm of \hyperref[sec:expensive]{Section~\ref*{sec:expensive}} requires an efficient algorithm to
report all the intersecting pairs of disks.
\begin{lemma}
    \label{lemma:sweepline}%
    Given a set $\mathcal{D}$ of $n$ disks, all the intersecting pairs of
    disks of $\mathcal{D}$ can be computed in $O(n\log n + k)$ expected
    time, where $k$ is the number of intersecting pairs.
\end{lemma}
\begin{proof}
    We break the boundary of each disk at its two $x$-extreme points,
    resulting in a set of $2n$ $x$-monotone curves.  Computing the
    vertical decomposition of the arrangement of these disks (curves)
    $\mathcal{A}\mleft({\mathcal{D}}\mright)$ can be done in
    $O(n \log n + k)$ expected time \cite{h-gaa-11}. See
    \hyperref[fig:vert-dec]{Figure~\ref*{fig:vert-dec}} for an
    example.  This gives us readily all the pairs that their
    boundaries intersect.

    As for the intersections that rise out of containment, perform a
    traversal of the dual graph of the vertical decomposition (i.e.,
    each vertical trapezoid is a vertex, and two trapezoids are
    adjacent if they share a boundary edge). The dual graph is planar
    with $O(n+k)$ vertices and edges, and as such the graph can be
    traversed in $O(n+k)$ time.  During the traversal, by appropriate
    bookkeeping, it is straightforward to maintain the list of disks
    containing the current trapezoid, in $O(1)$ per edge traversed, as
    any edge traversed changes this set by at most one.

    For a disk $\mathsf{d}$, let $p_{\mathsf{d}}$ be the rightmost point of
    $\mathsf{d}$. For each disk $\mathsf{d}$, pick any trapezoid $\Delta$ such
    that $p_{\mathsf{d}} \in \Delta$ and either the top boundary of $\mathsf{d}$
    is the ceiling of $\Delta$ or the bottom boundary of $\mathsf{d}$ is
    the floor of $\Delta$.  Assign $\Delta_{\mathsf{d}} \leftarrow \Delta$
    as the {\em representative} trapezoid of $\mathsf{d}$.

    During the traversal of the dual graph, consider the case where we
    arrive at a representative trapezoid of a disk $\mathsf{d}$. Let
    $L(\mathsf{d})$ be the list of disks containing $\Delta_{\mathsf{d}}$. Then
    scan $L(\mathsf{d})$ to report all the containing pairs of
    $\mathsf{d}$. Each disk in $L(\mathsf{d})$ either intersects the boundary of
    $\mathsf{d}$, or contain it. Therefore, the total time spent at the
    representative trapezoids is
    $\sum_{\mathsf{d} \in \mathcal{D}}\left| {L(d)} \right| = O(k)$.
\end{proof}

\begin{figure}[h]
    \centering%
    \includegraphics{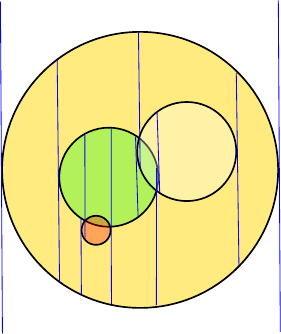}
    \caption{Vertical decomposition of four disks.}
    \label{fig:vert-dec}
\end{figure}

\section{From reporting to approximate sampling/counting}
\label{sec:app-sampling}

In this section, given a set of objects $\mathcal{D}$, and a reporting
data-structure for $\mathcal{D}$, we show a reduction to building a
data-structure, such that given a query object $\mathsf{q}$, it returns
approximately-at-uniform an object from
\begin{math}
    \mathsf{q} \sqcap \mathcal{D}%
    =%
    \left\{  x \in \mathcal{D} \;\middle\vert\; x \cap \mathsf{q} \neq \emptyset \right\},
\end{math}
and also returns an $(1\pm \varepsilon)$-approximation for the size of this
set.

\subsection{The data-structure}

\paragraph{The given reporting data-structure.}

Let $\mathcal{D}$ be a set of $n$ objects, and assume for any subset
$X \subseteq \mathcal{D}$ of size $m$, one can construct, in $C(m)$ time, a
data-structure that given a query object $\mathsf{d}$, returns, in
$O( Q(m) +k )$ time, all the objects in $X$ that intersects $\mathsf{d}$,
where $k = \left| {\mathsf{d} \sqcap X} \right|$.  Furthermore, we assume that if a
parameter $k'$ is specified by the query, then the data-structure
stops after $O( Q(m) +k')$ time, if $k >k'$, and indicate that this is
the case.

\begin{example}
    If $\mathcal{D}$ is a set of disks, and the query $\mathsf{d}$ is a disk,
    then this becomes a query reporting of the $k$-nearest neighbors
    in an additive weighted Voronoi diagram.  Liu \cite{l-nopkn-22}
    showed how to build such a reporting data-structure, in
    $O(n \log n)$ expected preprocessing time, and query time
    $O( \log n + k)$.
\end{example}

\paragraph{Data-structure construction.}

We build a random binary tree over the objects of $\mathcal{D}$, by
assigning each object of $\mathcal{D}$ with equal probability either a $0$
or a $1$ label. This partitions $\mathcal{D}$ into two sets $\mathcal{D}_0$
(label~$0$) and $\mathcal{D}_1$ (label~$1$).
Recursively build random trees $T_0$ and $T_1$ for $\mathcal{D}_0$ and
$\mathcal{D}_1$, respectively, with the output tree having $T_0$ and $T_1$
as the two children of the root. The constructions bottoms out when
the set of objects is of size one. Let $\mathcal{T}$ be the resulting tree
that has exactly $n$ leaves. For every node $u$ of $\mathcal{T}$, we
construct the reporting data-structure for the set of objects
$\mathcal{D}(u)$ -- that is, the set of objects stored in the subtree of
$u$.

Finally, create an array $L_i$ for each level $i$ of the tree $\mathcal{T}$,
containing pointers to all the nodes in this level of the tree.

\paragraph{Answering a query.}

Given a query object $\mathsf{q}$, and a parameter $\varepsilon \in (0,1/2)$, the
algorithm starts from an arbitrary leaf $v$ of $\mathcal{T}$. The leaf $v$
has a unique root to $v$ path, denoted by $\pi = u_0 u_1 \ldots u_t $,
where $u_0 = \mathrm{root}(\mathcal{T})$ and $u_t = v$. The algorithm
performs a binary search on $\pi$, using the reporting data-structure
associated with each node, to find the maximal $i$, such that
$\left| {\mathsf{q} \sqcap \mathcal{D}(u_i)} \right| > \psi = c\log n$, where $c$ is a
sufficiently large constant. Here, we use the property that one can
abort the reporting query if the number of reported objects exceeds
$\psi$. This implies that each query takes $Q(n) + O( \log n)$ time
(with no dependency on $\varepsilon$). Next, the algorithm computes the
maximal $j$ such that
\begin{math}
    \left| {\mathsf{q} \sqcap \mathcal{D}(u_j)} \right| > \psi_\varepsilon = c\varepsilon^{-2}\log n
\end{math}
(or set $j=0$ if such a $j$ does not exist).  This is done by going up
the path $\pi$ from $u_i$, trying $u_{i-1}, u_{i-2}, \ldots, u_j$
using the reporting data-structure till the condition is
fulfilled\footnote{One can also ``jump'' to level
   $i+\log_2(1/\varepsilon^2)-2$, and do a local search there for $j$, but
   this ``improvement'' does not effect the performance.}.  Next, the
algorithm chooses a vertex $u \in L_j$ uniformly at random. It
computes the set $S= \mathsf{q} \sqcap \mathcal{D}(u)$ using the reporting
data-structure. The algorithm then returns a random object from $S$
uniformly at random, and the number $2^j \left| {S} \right|$. The first is a
random element chosen from $\mathsf{q} \sqcap \mathcal{D}$, and the second
quantity returned is an estimate for $\left| {\mathsf{q} \sqcap \mathcal{D}} \right|$.

\subsection{Analysis}

\subsubsection{Correctness}
\begin{lemma}
    \label{lemma:a:high:is:good}%
    Let $\varepsilon \in (0,1/2)$, $\psi_\varepsilon = c\varepsilon^{-2}\log n$, and
    $\mathsf{q}$ be a query object.  Let $M \geq 1$ be the integer such
    that
    \begin{math}
        \psi_\varepsilon /16 \leq \left| { \mathcal{D} \sqcap \mathsf{q}} \right| /2^{M} <
        \psi_\varepsilon /8.
    \end{math}
    Then, for all nodes $v$ at distance $i \leq M$ from the root of
    $\mathcal{T}$, we have
    \begin{math}
        {\mathbb{P}}\mleft[ \Bigl.\smash{\left| {\mathcal{D}(v) \sqcap \mathsf{q}} \right| \notin
              (1\pm\varepsilon/2)\frac{\left| {\mathcal{D} \sqcap \mathsf{q}} \right|}{2^i}} \mright]
        \leq%
        \frac{1}{n^{\Omega(c)}}.
    \end{math}
\end{lemma}
\begin{proof}Consider a node $v$ at a distance $i$ from the root, and
    let $Y_v = \left| {\mathcal{D}(v) \sqcap \mathsf{q}} \right|$.  Clearly,
    $\mu_v = {\mathbb{E}}\mleft[ Y_v \mright] = \left| { \mathcal{D} \sqcap \mathsf{q}} \right|/2^i$.  Since
    $i\leq M$, we have $\mu_v \geq \psi_\varepsilon/16$.  By
    \hyperref[theo:special:h:i:e]{Chernoff's inequality}, we have
    \begin{equation*}
        {\mathbb{P}}\mleft[ \bigl. Y_v \notin (1\pm \varepsilon/2)\mu_v \mright]
        \leq%
        2\exp\mleft({ - \varepsilon^2 \mu_v /3 }\mright)
        \leq%
        2\exp\mleft({ - \varepsilon^2 \psi_\varepsilon /48 }\mright)
        \leq
        2\exp\mleft({ - (c/48)\log n }\mright)
        \leq
        \frac{2}{n^{c/48}}.
    \end{equation*}
    The number of nodes in $\mathcal{T}$ is $O(n)$, and hence, by the union
    bound, for all nodes $v$ at distance $i\leq M$ from the root, we
    have
    ${\mathbb{P}}\mleft[ Y_v \notin (1\pm \varepsilon/2)\mu_v \mright] \leq {1}/{n^{\Omega(c)}}$.
\end{proof}

\begin{observation}
    \label{observation:low:is:small}%
    \hyperref[lemma:a:high:is:good]{Lemma~\ref*{lemma:a:high:is:good}}
    implies that for all nodes $v$ at distance $i > M$ from the root
    of $\mathcal{T}$, we have
    \begin{math}
        {\mathbb{P}}\mleft[ \bigl.\left| {\mathcal{D}(v) \sqcap \mathsf{q}} \right| > \psi_{\varepsilon} \mright]
        \leq {1}/{n^{\Omega(c)}}.
    \end{math}
    Indeed, \hyperref[lemma:a:high:is:good]{Lemma~\ref*{lemma:a:high:is:good}} implies this for all nodes at
    distance $M$ from the root, and the sizes of these sets are
    monotonically decreasing along any path down the tree.
\end{observation}

\begin{lemma}
    Assume that the number of distinct sets $\mathsf{q}' \sqcap \mathcal{D}$,
    over all possible query objects $\mathsf{q}'$, is bounded by a
    polynomial $O(n^d)$, where $d$ is some constant. Then, for a query
    $\mathsf{q}$, the probability that the algorithm returns a specific
    object $\mathsf{o} \in \mathcal{D} \sqcap \mathsf{q}$, is in $(1\pm \varepsilon)/ \beta$,
    where $\beta = \left| {\mathcal{D} \sqcap \mathsf{q}} \right|$. Similarly, the
    estimate the algorithm outputs for $\beta$ is in
    $(1\pm\varepsilon) \beta$. The answer is correct for all queries, with
    probability $\geq 1- 1/n^{\Omega(c)}$, for a sufficiently large
    constant $c$.
\end{lemma}
\begin{proof}
    It is easy to verify the algorithm works correctly if
    $\left| {\mathsf{q} \sqcap \mathcal{D}(u_j)} \right| < \psi_\varepsilon$.  Otherwise, for
    the node $u_j$ computed by the algorithm, we have
    $\left| {\mathsf{q} \sqcap \mathcal{D}(u_j)} \right| > \psi_\varepsilon = c\varepsilon^{-2}\log
    n$. By \hyperref[observation:low:is:small]{Observation~\ref*{observation:low:is:small}}, with high probability, we have that
    $j\leq M$. By \hyperref[lemma:a:high:is:good]{Lemma~\ref*{lemma:a:high:is:good}}, it implies that for any
    node $u \in L_j$, we have
    $2^j\left| {\mathcal{D}(u) \sqcap \mathsf{q}} \right| \in (1\pm\varepsilon/2)\left| {\mathcal{D}
       \sqcap \mathsf{q}} \right| = (1\pm \varepsilon/2)\beta$, which implies that the
    estimate for the size of $\beta$ is correct, as $u \in L_j$. This
    readily implies that the probability of returning a specific
    object $\mathsf{o} \in \mathcal{D} \sqcap \mathsf{q}$ is in $(1\pm \varepsilon)/\beta$,
    since
    \begin{equation*}
        \frac{1-\varepsilon}{\beta}
        \leq
        \frac{1}{ (1+\varepsilon/2)  \left| {\mathcal{D}
              \sqcap \mathsf{q}} \right|}
        \leq
        \frac{1}{\left| {L_j} \right| \cdot \left| {\mathcal{D}(u)
              \sqcap \mathsf{q}} \right|}
        \leq
        \frac{1}{ (1-\varepsilon/2)  \left| {\mathcal{D}
              \sqcap \mathsf{q}} \right|}
        \leq
        \frac{1+\varepsilon}{\beta}.
    \end{equation*}

    As for the probabilities, there are $n$ nodes in $\mathcal{T}$, and
    $O(n^d)$ different queries, and thus the probability of failure is
    at most $n^{d+1} /n^{\Omega(c)} < 1 / n^{\Omega(c)}$, by
    \hyperref[lemma:a:high:is:good]{Lemma~\ref*{lemma:a:high:is:good}}.
\end{proof}

\subsubsection{Running times}

\paragraph{Query time.}
The depth of $\mathcal{T}$ is $h = O( \log n)$ with high probability
(follows readily from Chernoff's inequality). Thus, the first stage
(of computing the maximal $i$) requires $O(\log \log n)$ queries on
the reporting data-structure, where each query takes
$O( Q(n) + \log n)$ time. The second stage (of finding maximal $j$)
takes
\begin{equation*}
    \tau
    =%
    {\mathbb{E}}\mleft[ \Bigl.\smash{\sum\nolimits_{t=j}^i} O( Q(n) + \left| {\mathcal{D}( u_t)
          \sqcap \mathsf{q}} \right|) \mright].
\end{equation*}
Thus, we have
\begin{compactenumA}
    \smallskip%
    \item If $Q(n) = O(\log n)$, then $\tau = O( \psi_\varepsilon)$, as the
    cardinality of $\mathcal{D}( u_t) \sqcap \mathsf{q}$ decreases by a factor
    of two (in expectation) as one move downward along a path in the
    tree. Thus $\tau$ is a geometric summation in this case dominated
    by the largest term.

    \smallskip%
    \item If $Q(n)=\Omega(\log n)$, we have (in expectation) that
    $\left| { i - j } \right| \leq O( \log (1/\varepsilon))$, and thus
    $\tau = O( Q(n) \log (1/\varepsilon) + \psi_\varepsilon )$ time.

    \smallskip%
    \item If $Q(n)=O(n^{\lambda})$, for $0< \lambda \leq 1$, then the
    query time is dominated by the query time for the top node (i.e.,
    $u_j$) in this path, and $\tau = O(Q(n))$, as can easily be
    verified.
\end{compactenumA}

\paragraph{Construction time.}
The running time bounds of the form $O(C(n))$ are
\emphw{well-behaved}, if for any non-negative integers
$n_1, n_2, \ldots$, such that $\sum_{i=1}n_i=n$, implies that
$\sum_{i=1} C(n_i)=O(C(n))$.  Under this assumption on the
construction time, we have that the total construction time is
$O( C(n) \log n)$.

\subsubsection{Summary}

\begin{theorem}%
    {\label{theo:gen-red-I}}%
    Let $\mathcal{D}$ be a set of $n$ objects, and assume we are given a
    well-behaved range-reporting data-structure that can be
    constructed in $C(m)$ time, for $m$ objects, and answers a
    reporting query $\mathsf{q}$ in
    $O( Q(m) + \left| {\mathsf{q} \sqcap \mathcal{D}} \right| )$ time. Then, one can
    construct a data-structure, in $O( C(n ) \log n)$ time, such that
    given a query object $\mathsf{q}$, it reports an $(1\pm\varepsilon)$-estimate
    for $\beta = \left| {\mathsf{q} \sqcap \mathcal{D}} \right|$, and also returns an
    object from $\mathsf{q} \sqcap \mathcal{D}$, where each object is reported
    with probability $(1 \pm \varepsilon)/\beta$. The data-structure answers
    all such queries correctly with probability
    $\geq 1 - 1/n^{\Omega(1)}$.  The expected query time is:
    \begin{compactenumi}
        \smallskip%
        \item $O( (\varepsilon^{-2} + \log \log n) \log n)$ if
        $Q(m) = O( \log m)$.

        \smallskip%
        \item $O( Q(n))$ if $Q(m) =O(m^{\lambda})$, for some constant
        $\lambda > 0$.

        \smallskip%
        \item
        $O\bigl( \varepsilon^{-2} \log n + Q(n) \log \tfrac{\log n}{\varepsilon} )$
        otherwise.
    \end{compactenumi}
\end{theorem}

Plugging in the data-structure of Liu \cite{l-nopkn-22} for disks,
with $C(m)=O(m\log m)$, and $Q(m)=O(\log m)$, in the above theorem,
implies the following.

\begin{corollary}\label{cor:disks-estimate}
    Let $\mathcal{D}$ be a set of $n$ disks in the plane.  One can construct
    in $O(n \log^2 n)$ time a data-structure, such that given a query
    disk $\mathsf{q}$ and a parameter $\varepsilon \in (0,1/2)$, it outputs an
    $(1\pm\varepsilon)$-estimate for $\beta = \left| {\mathsf{q} \sqcap \mathcal{D}} \right|$,
    and also returns a disk in $\mathsf{q} \sqcap \mathcal{D}$ with a
    probability that is $(1\pm \varepsilon)$-uniform.  The expected query
    time is $O( (\varepsilon^{-2} + \log \log n) \log n)$, and the result
    returned is correct with high probability for all possible
    queries.
\end{corollary}

\section{A $(2+\varepsilon)$-approximation for densest subset disks}

In this section we design a $(2+\varepsilon)$-approximation algorithm to
compute the densest subset of disks in
$O(n\varepsilon^{-2}\log^3n + n\varepsilon^{-4}\log^2n)$ time.

\subsection{The algorithm}

The input is a set $\mathcal{D}$ of $n$ disks in the plane.  Let
$\vartheta = \varepsilon/15$.  The basic idea is to try a sequence of
exponentially decaying values to the optimal density $\mathcalb{d}$. To this
end, in the $i$th\xspace round, the algorithm would try the degree threshold
$\beta=n(1-\vartheta)^i$.

In the beginning of such a round, let $\mathcal{L} \leftarrow \mathcal{D}$.  During
a round, the algorithm repeatedly removes ``low-degree'' objects, by
repeatedly doing the following:
\begin{compactenumI}[itemsep=-0.5ex]
    \smallskip%
    \item The algorithm constructs the data-structure of
    \hyperref[cor:disks-estimate]{Corollary~\ref*{cor:disks-estimate}} on the objects of $\mathcal{L}$.

    \item Let $\mathcal{L}_< \subseteq \mathcal{L}$ be the objects whose degree in
    $\mathcal{L}$ is smaller than $(1+\vartheta)\beta$ according to this
    data-structure. Let $\mathcal{L}_\geq = \mathcal{L} \setminus \mathcal{L}_<$.

    \item If $\mathcal{L}_\geq$ is empty, then this round failed, and the
    algorithm continues to the next round.

    \item If $\left| {L_{<}} \right| < \vartheta \left| {\mathcal{L}} \right|$, then the algorithm
    returns $\mathcal{L}$ as the desired approximate densest subset.

    \item Otherwise, let $\mathcal{L} \leftarrow \mathcal{L} \setminus L_{<}$.  The
    algorithm continues to the next iteration (still inside the same
    round).
\end{compactenumI}

\subsection{Analysis}

\begin{lemma}
    \label{lemma:beta-at-least}%
    When the algorithm terminates, we have
    $\beta \geq (1-\vartheta)^3\mathcalb{d}$, with high probability, where
    $\mathcalb{d}$ is the optimal density.
\end{lemma}
\begin{proof}
    Consider the iteration when
    $\beta \in [(1-\vartheta)^3\mathcalb{d}, (1-\vartheta)^2\mathcalb{d}]$. By definition of
    $\mathcal{L}_<$, all the objects in it have a degree at most
    $(1+\vartheta)\beta \leq (1+\vartheta)(1-\vartheta)^2\mathcalb{d}\leq
    (1-\vartheta)\mathcalb{d}$.  By \hyperref[lemma:density:l:b]{Lemma~\ref*{lemma:density:l:b}}, all the objects in the
    optimal solution have degree $\geq \mathcalb{d}$ (when restricted to the
    optimal solution).  Therefore, none of the objects in the optimal
    solution are in $\mathcal{L}_<$ and hence, the set $\mathcal{L}_\geq$ is not empty
    (and contains the optimal solution). Inside this round, the loop
    is performed at most $O( \vartheta^{-1} \log n)$ times, as every
    iteration of the loop shrinks $\mathcal{L}$ by a factor of $1-\vartheta$. This
    implies that the algorithm must stop in this round.
\end{proof}
\begin{lemma}
    \label{lemma:edges-remaining}%
    The above algorithm returns a $(2+\varepsilon)$-approximation of the
    densest subset.
\end{lemma}
\begin{proof}
    Consider the set $\mathcal{L}$, the value of $\beta$ when the algorithm
    terminated, and let $\nu =\left| {\mathcal{L}} \right|$.  By
    \hyperref[lemma:beta-at-least]{Lemma~\ref*{lemma:beta-at-least}}, $\beta \geq (1-\vartheta)^3\mathcalb{d}$, and by the
    algorithm stopping condition, we have
    \begin{math}
        \left| {\mathcal{L}_<} \right| < \vartheta \left| {\mathcal{L}} \right|.
    \end{math}
    In addition, all the objects in $\mathcal{L}_\geq$ have degree at least
    $(1-\vartheta) \beta$. Thus, the number of induced edges on $\mathcal{L}$ is
    \begin{equation*}
        \mathsf{m}\mleft({\mathcal{L}}\mright)
        \geq
        \frac{(1-\vartheta) \beta \left| {\mathcal{L}_\geq} \right|}{2}
        >
        \frac{(1-\vartheta)^2 \beta \left| {\mathcal{L}} \right|}{2}
        \geq
        (1-\vartheta)^5\frac{\mathcalb{d}}{2} \left| {\mathcal{L}} \right|
        \geq
        (1-5\vartheta) \frac{\mathcalb{d}}{2} \left| {\mathcal{L}} \right|.
    \end{equation*}
    Thus
    ${\nabla}\mleft({\mathcal{L}}\mright) = \frac{\mathsf{m}\mleft({\mathcal{L}}\mright)}{\left| {\mathcal{L}} \right|} \geq (1-5 \vartheta)
    \mathcalb{d}/2 = (1-\varepsilon/3)\mathcalb{d}/2$, as $\vartheta= \varepsilon/15$.  Observe that
    $2/(1-\varepsilon/3) \leq 2(1+\varepsilon/2) \leq 2 + \varepsilon$.
\end{proof}

\begin{theorem}
    {\label{theo:first}}%
    Let $\mathcal{D}$ be a set of $n$ disks in the plane, and let
    $\varepsilon \in (0,1)$ be a parameter.  The above algorithm computes, in
    $O(n\varepsilon^{-2}\log^3n + n\varepsilon^{-4}\log^2n)$ expected time, a
    $(2+\varepsilon)$-approximation to the densest subgraph of
    $\mathsf{G}_{\cap \mathcal{D}}$.  The result returned is correct, as is the running
    time bound, with high probability.
\end{theorem}
\begin{proof}
    The expected time taken in Step-I is $O(n\log^2n)$ and the
    expected time taken in Step-II is
    $O(n(\vartheta^{-2}+\log\log n)\log n)$.  As such, the time taken
    perform the partition step on $n$ disks is
    $O(n\log^2n + n\vartheta^{-2}\log n)$ expected time.

    For a fixed value of $\beta$, let $t$ be the number of times the
    partition step (i.e., step~(V)) is performed, and let
    $n_1,n_2,\ldots,n_t$ be the number of objects participating in the
    partition step. Clearly, $n_1=n$, and in general
    $n_i \leq (1-\vartheta)n_{i-1}, \forall i\in [2,t]$.  Therefore,
    $\sum_{i=1}^{t}n_i=O(n/\vartheta)$ and hence, the expected time taken
    to perform the $t$ partition steps is
    $\sum_{i=1}^{t}O(n_i\log^2n_i + n_i\vartheta^{-2}\log n_i)=
    O(n\vartheta^{-1}\log^2n + n\vartheta^{-3}\log n)$.  The value of $\beta$
    can change $O(\vartheta^{-1}\log n)$ times during the algorithm and
    hence, the overall expected running time of the algorithm is
    $O(n\vartheta^{-2}\log^3n + n\vartheta^{-4}\log^2 n)$.
\end{proof}

\section{An $(1+\varepsilon)$-approximation for densest subset disks}
\label{sec:expensive}

Here, we present a $(1+\varepsilon)$-approximation algorithm for the densest
subset of disks, which is based on the following intuitive idea -- if
the intersection graph is sparse, then the problem is readily
solvable. If not, then one can sample a sparse subgraph, and use an
approximation algorithm on the sampled graph.

\subsection{Densest subgraph estimation via sampling}

Let $\mathsf{G}=(\mathsf{V},\mathsf{E})$ be a graph with $n$ vertices and $m$ edges, with
maximum subgraph density $\mathcalb{d}$. Let $\vartheta \in (0,1/6)$ be a
parameter, and assume that $m > c'n\vartheta^{-2}\log n$, where $c'$ is
some sufficiently large constant, which in particular implies that
\begin{equation*}
    \mathcalb{d}
    =
    \mathcalb{d}\mleft({\mathsf{G}}\mright)
    \geq
    \frac{m}{n}
    \geq
    \frac{c' }{\vartheta^{2}} \log n.
\end{equation*}
Assume we have an estimate $\overline{m} \in (1\pm \vartheta) m$ of $m$.  For a
constant $c$ to be specified shortly, with $c < c'$, let
\begin{equation*}
    \psi
    =
    c\frac{n}{\overline{m}}\vartheta^{-2} \log n
    \leq
    \frac{c }{c'(1-\vartheta) }
    \leq
    \frac{6c}{5c'}
    <%
    1.
\end{equation*}
Let $F = \{ e_1, \ldots, e_r\}$ be a random sample of
$r = \left\lceil {\psi \overline{m}} \right\rceil$ edges from $\mathsf{G}$. Specifically, in the
$i$th\xspace iteration, an edge $e_i$ is picked from the graph, where the
probability of picking any edge is in $(1\pm \vartheta)/m$.  Let
$H = (V, F)$, and observe that $H$ is a sparse graph with
$n$ vertices and $r = O( \vartheta^{-2} n \log n)$ edges. The claim is
that the densest subset $D \subseteq V$ in $H$, or even approximate
densest subset, is close to being the densest subset in $\mathsf{G}$. The
proof of this follows from previous work \cite{mtvv-dsdgs-15}, but
requires some modifications, since we only have an estimate to the
number of edges $m$, and we are also interested in approximating the
densest subgraph on the resulting graph. We include the details here
so that the presentation is self contained. The result we get is
summarized in \hyperref[lemma:sampling]{Lemma~\ref*{lemma:sampling}}, if the reader is uninterested in the
(somewhat tedious) analysis of this algorithm.

\subsubsection{Analysis}

\begin{lemma}
    \label{lemma:small:large:a}%
    Let $\mathcalb{d} = \mathcalb{d}\mleft({\mathsf{G}}\mright)$, and let $U \subseteq V$, be an arbitrary
    set of $k$ vertices.  If $\nabla_{\!\!\mathsf{G}}\mleft({U}\mright) \leq \mathcalb{d}/60$, then
    ${\mathbb{P}}\mleft[ \nabla_{\!\!H}\mleft({U}\mright) \geq \psi \mathcalb{d} /5  \mright] \leq n^{-100k}$.
\end{lemma}
\begin{proof}
    We have $\mathcalb{d} \geq m/ n$, where $n = \left| {\mathsf{V}\mleft({\mathsf{G}}\mright)} \right|$ and
    $m = \left| {\mathsf{E}\mleft({ \mathsf{G}}\mright)} \right|$, and thus
    \begin{equation*}
        \psi
        =
        c \frac{n}{\overline{m} \vartheta^2 } \log n
        \geq
        c \frac{n}{(1+\vartheta)m \vartheta^2 } \log n
        \geq
        \frac{1}{\mathcalb{d}} \cdot  \frac{c}{(1+\vartheta) \vartheta^2 } \log n
    \end{equation*}
    Let $X_i=1$ if the edge sampled in the $i$th\xspace round belongs to
    $H_{U}$ and zero, otherwise.  Let $X=\sum_{i} X_i$ be the
    number of edges in $H_{U}$.  Then
    \begin{equation*}
        {\mathbb{P}}\mleft[ X_i=1 \mright]
        \in%
        (1\pm\vartheta)\frac{\left| {\mathsf{E}\mleft({\mathsf{G}_{U}}\mright)} \right|}{m}
        =%
        (1\pm\vartheta) \frac{k \nabla_{\!\!\mathsf{G}}\mleft({U}\mright)}{m}.
    \end{equation*}
    By linearity of expectations, and as $\overline{m} \in (1\pm \vartheta) m$,
    we have
    \begin{equation}
        {\mathbb{E}}\mleft[ X \mright]%
        \in%
        (1\pm\vartheta) \psi \overline{m} \frac{k \nabla_{\!\!\mathsf{G}}\mleft({U}\mright)}{m}%
        \subseteq%
        (1\pm\vartheta)^2 \psi k \nabla_{\!\!\mathsf{G}}\mleft({U}\mright).
        \label{equation:interval}%
    \end{equation}
    By assumption $\nabla_{\!\!\mathsf{G}}\mleft({U}\mright) \leq \mathcalb{d}/60$, implying that
    ${\mathbb{E}}\mleft[ X \mright] \leq (1+3\vartheta) \psi k \mathcalb{d}/60 \leq \psi k \mathcalb{d}/30$, if
    $\vartheta \in (0,1/3)$, by \hyperref[observation:eps:games]{Observation~\ref*{observation:eps:games}}.  Observe that
    \begin{math}
        (1+ (2e-1)){\mathbb{E}}\mleft[ X \mright] \leq \frac{\psi k \mathcalb{d}}{5}.
    \end{math}
    By Chernoff's inequality, \hyperref[lemma:chernoff:delta:big]{Lemma~\ref*{lemma:chernoff:delta:big}}, we have

    \begin{align*}
      {\mathbb{P}}\mleft[ \nabla_{\!\!H}\mleft({U}\mright) \geq \frac{\psi \mathcalb{d}}{5} \mright]
      &=
        {\mathbb{P}}\mleft[ X \geq \frac{\psi k \mathcalb{d}}{ 5} \mright]
        \leq%
        2^{-\psi k \mathcalb{d}/5}
        \leq
        \frac{1}{n^{100k}},
    \end{align*}
    by picking $c$ to be sufficiently large.
\end{proof}

\begin{lemma}
    \label{lemma:small:large:b}%
    Let $\mathcalb{d} = \mathcalb{d}\mleft({\mathsf{G}}\mright)$, and let $U \subseteq V$, be an arbitrary
    set of $k$ vertices.  If $\nabla_{\!\!\mathsf{G}}\mleft({U}\mright) \geq \mathcalb{d}/60$, then
    ${\mathbb{P}}\mleft[ \nabla_{\!\!H}\mleft({U}\mright) \in (1\pm \vartheta)^3 \psi \nabla_{\!\!\mathsf{G}}\mleft({U}\mright)  \mright]
    \geq 1-n^{-100k}$ .
\end{lemma}
\begin{proof}
    Following the argument of \hyperref[lemma:small:large:a]{Lemma~\ref*{lemma:small:large:a}} and as
    $\nabla_{\!\!\mathsf{G}}\mleft({U}\mright) \geq \mathcalb{d}/60$, we have that
    ${\mathbb{E}}\mleft[ X \mright] = \Omega(k\cdot \psi \nabla_{\!\!\mathsf{G}}\mleft({U}\mright))= \Omega(k\cdot \psi
    \mathcalb{d}) = \Omega( k \cdot c \vartheta^{-2} \log n)$. Chernoff's
    inequality, \hyperref[theo:special:h:i:e]{Theorem~\ref*{theo:special:h:i:e}}, then implies that
    $X \in (1\pm \vartheta) {\mathbb{E}}\mleft[ X \mright]$ with probability at least
    $1-2\exp( - \vartheta^2 {\mathbb{E}}\mleft[ X \mright]/4) \geq 1 - 1/n^{100k}$, for $n$
    sufficiently large. The claim now readily follows from
    \hyperref[equation:interval]{Eq.~(\ref*{equation:interval}{)}}.
\end{proof}

\begin{lemma}
    \label{lemma:happy}%
    Let $\alpha\in(0,1/6)$ be a parameter. For all sets
    $U \subseteq V$, such that
    $\nabla_{\!\!H}\mleft({U}\mright) \geq (1-\alpha)\mathcalb{d}\mleft({H}\mright)$, we have that
    $\nabla_{\!\!\mathsf{G}}\mleft({U}\mright) \geq (1-6\vartheta)(1-\alpha)\mathcalb{d}$, and this holds
    with high probability.
\end{lemma}
\begin{proof}
    Let $X$ be the densest subset in $\mathsf{G}$. By \hyperref[lemma:small:large:b]{Lemma~\ref*{lemma:small:large:b}},
    we have that
    \begin{equation*}
        \nabla_{\!\!H}\mleft({X}\mright) \in (1\pm \vartheta)^3 \psi \mathcalb{d}\mleft({\mathsf{G}}\mright)
        \implies
        \mathcalb{d}\mleft({H}\mright) \geq  (1 - \vartheta)^3 \psi \mathcalb{d}
        \geq
        \frac{\psi \mathcalb{d}}{2}.
    \end{equation*}

    By \hyperref[lemma:small:large:a]{Lemma~\ref*{lemma:small:large:a}}, we have that for all the sets
    $T \subseteq \mathsf{V}$, with $\nabla_{\!\!\mathsf{G}}\mleft({T}\mright) \leq \mathcalb{d}/60$, we have
    $\nabla_{\!\!H}\mleft({T}\mright) < \psi \mathcalb{d} /5 < \mathcalb{d}\mleft({H}\mright)/2$, and this
    happens with probability
    $\sum_{k=2}^n \sum_{T \subseteq \mathsf{V}: \left| {T} \right|=k} 1/n^{100k} \leq
    \sum_{k=2}^n \binom{n}{k}/n^{100k} \leq 1/n^{99}$.

    Thus, all the sets $U \subseteq \mathsf{V}$ under consideration have
    $\nabla_{\!\!\mathsf{G}}\mleft({U}\mright) > \mathcalb{d}/60$.  By \hyperref[lemma:small:large:b]{Lemma~\ref*{lemma:small:large:b}}, for all
    such sets, with probability $1-n^{-100k}\geq 1-1/n^{99}$, we have
    $\nabla_{\!\!H}\mleft({U}\mright) \in (1\pm \vartheta)^3 \psi \nabla_{\!\!\mathsf{G}}\mleft({U}\mright)$, which
    implies
    \begin{math}
        \nabla_{\!\!\mathsf{G}}\mleft({U}\mright) \in \frac{1}{(1\pm \vartheta)^3 \psi}
        \nabla_{\!\!H}\mleft({U}\mright).
    \end{math}
    Thus, we have
    \begin{equation*}
        \nabla_{\!\!\mathsf{G}}\mleft({U}\mright)
        \geq
        \frac{1}{(1+\vartheta)^3 \psi}
        \nabla_{\!\!H}\mleft({U}\mright)
        \geq
        \frac{(1-\alpha)\mathcalb{d}\mleft({H}\mright)}{(1+\vartheta)^3 \psi}
        \geq
        \frac{(1-\alpha)(1 - \vartheta)^3 \psi \mathcalb{d}}{(1+\vartheta)^3 \psi}
        \geq
        (1-\alpha)(1-6\vartheta) \mathcalb{d},
    \end{equation*}
    since $1/(1+\vartheta) \geq 1-\vartheta$, and $(1-\vartheta)^6 \geq 1-6\vartheta$.
\end{proof}

\subsubsection{Summary}

\begin{lemma}
    \label{lemma:sampling}%
    Let $\varepsilon \in (0,1)$ be a parameter, and let $\mathsf{G}=(\mathsf{V},\mathsf{E})$ be a
    graph with $n$ vertices and $m$ edges, with
    $ m = \Omega(\varepsilon^{-2} n \log n)$. Furthermore, let $\overline{m}$ be an
    estimate to $m$, such that $\overline{m} \in (1\pm \vartheta)m$, where
    $\vartheta = \varepsilon/10$.  Let $\psi = c(n/\overline{m})\vartheta^{-2} \log n$, and
    let $F$ be a random sample of
    $\psi \overline{m} = O( \varepsilon^{-2} n\log n)$ edges, with repetition, where
    the probability of any specific edge to be picked is
    $(1\pm \vartheta)/ m$, and $c$ is a sufficiently large constant. Let
    $H = (\mathsf{V}, F)$ be the resulting graph, and let
    $X \subseteq \mathsf{V}$ be subset of $H$ with
    $\nabla_{\!\!H}\mleft({X}\mright)\geq(1-\varepsilon/6)\mathcalb{d}\mleft({H}\mright)$. Then,
    ${\nabla}\mleft({X}\mright) \geq (1-\varepsilon)\mathcalb{d}\mleft({\mathsf{G}}\mright)$.
\end{lemma}
\begin{proof}
    This follows readily from the above, by setting $\alpha =\varepsilon/6$,
    and using \hyperref[lemma:happy]{Lemma~\ref*{lemma:happy}}.
\end{proof}

\subsection{Random sampler}

To implement the above algorithm, we need an efficient algorithm for
sampling edges from the intersection graph of disks, which we describe
next.

\subsubsection{The algorithm}
The algorithm consists of the following steps:
\begin{compactenumI}[itemsep=-0.5ex]
    \item Build the data-structure of \hyperref[cor:disks-estimate]{Corollary~\ref*{cor:disks-estimate}} on the
    disks of $\mathcal{D}$ with error parameter $\varepsilon/c$, where $c$ is a
    sufficiently large constant. Also, build the range-reporting
    data-structure of Liu~\cite{l-nopkn-22} on the disks of $\mathcal{D}$.

    \item For each object $\mathsf{o} \in \mathcal{D}$, query the data-structure
    of \hyperref[cor:disks-estimate]{Corollary~\ref*{cor:disks-estimate}} with $\mathsf{o}$. Let the estimate returned
    be $d'$. If $d' < c'/\varepsilon$ (for a constant $c' \gg c$), then
    report $\mathsf{o} \sqcap \mathcal{D}$ by querying the range-reporting
    data-structure with $\mathsf{o}$, and set
    $d_{\mathsf{o}}\leftarrow |\mathsf{o} \sqcap \mathcal{D}|-1$.  Otherwise, set
    $d_{\mathsf{o}} \leftarrow d'$.

    \item We perform $|F|$ iterations and in each iteration,
    sample a random edge from $\mathsf{G}_{\cap \mathcal{D}}$. In a given iteration,
    sample a disk $\mathsf{o} \in \mathcal{D}$, where $\mathsf{o}$ has a probability of
    $\frac{d_{\mathsf{o}}}{\sum_{\mathsf{o} \in \mathcal{D}}d_{\mathsf{o}}}$ being sampled.
    If $d_{\mathsf{o}} < c'/\varepsilon$, then uniformly-at-random report a disk
    from $\mathsf{o} \sqcap (\mathcal{D}-\mathsf{o})$.  Otherwise, query the
    data-structure of \hyperref[cor:disks-estimate]{Corollary~\ref*{cor:disks-estimate}} with $\mathsf{o}$ which
    returns a disk in $\mathsf{o} \sqcap \mathcal{D}$ (keep querying till a disk
    other than $\mathsf{o}$ is returned).
\end{compactenumI}

\subsubsection{Analysis}

\begin{lemma}%
    \label{lemma:approx-degree}%
    For each object $\mathsf{o} \in \mathcal{D}$, we have
    $d_{\mathsf{o}}\in (1\pm \varepsilon/c'')|\mathsf{o} \sqcap (\mathcal{D}-\mathsf{o})|$, where
    $c \gg c''$ with high probability.
\end{lemma}
\begin{proof}
    Fix an object $\mathsf{o}$. When $d' < c'/\varepsilon$, then the statement
    holds trivially.  Let ${d}=|\mathsf{o} \sqcap (\mathcal{D}-\mathsf{o})|$. Now we
    consider the case $d' \geq c'/\varepsilon$.  We know that
    $d' \in (1\pm \varepsilon/c)({d}+1)$.  Firstly,
    $d' \leq (1+\varepsilon/c)({d}+1)\leq (1+\varepsilon/c''){d}$ hold if
    \begin{equation*}
        {d} \geq \frac{1}{\varepsilon}\cdot\frac{1+ \varepsilon/c}{1/c''-1/c} \geq 1/2\varepsilon.
    \end{equation*}
    Observe that ${d} \geq 1/2\varepsilon$, since
    $c'/\varepsilon \leq d' \leq (1+\varepsilon/c)({d}+1)$ implies that
    ${d} \geq \frac{c'}{\varepsilon(1+\varepsilon/c)}-1 \geq 1/2\varepsilon$.  Finally,
    $d'\geq (1-\varepsilon/c)({d}+1) \geq (1-\varepsilon/c''){d}$ holds
    trivially. Therefore,
    \begin{equation*}
        d_{\mathsf{o}}=d' \in (1\pm \varepsilon/c''){d}.
    \end{equation*}
\end{proof}

\begin{lemma}
    In each iteration, the probability of sampling any edge in
    $\mathsf{G}_{\cap \mathcal{D}}$ is $(1\pm \varepsilon)/m$.
\end{lemma}
\begin{proof}
    An edge $(u,v)$ in $\mathsf{G}_{\cap \mathcal{D}}$ can get sampled in step~(III) in
    two ways. In the first way, the disk corresponding to $u$ gets
    sampled and then $v$ gets reported as the random neighbor of $u$,
    and vice-versa for the second way.

    Let $d=|\mathsf{o} \sqcap (\mathcal{D}-\mathsf{o})|$, where $\mathsf{o}$ is the disk
    corresponding to $u$. Consider the case, where
    $d_{\mathsf{o}} \geq c'/\varepsilon$.  As such, the first way of sampling the
    edge $(u,v)$ has probability lower-bounded by:
    \begin{equation*}
        \frac{d_{\mathsf{o}}}{\sum d_{\mathsf{o}}}\cdot \underbrace{\frac{1\pm\varepsilon/c}{d}}_{\text{almost-uniformity}} \subseteq
        \frac{d_{\mathsf{o}}}{\underbrace{(1 \pm \varepsilon/c'')2m}_{\hyperref[lemma:approx-degree]{Lemma~\ref*{lemma:approx-degree}}}}\cdot
        \frac{1 \pm \varepsilon/c}{d} \subseteq
        \frac{(1\pm \varepsilon/c'')d}{(1\pm \varepsilon/c'')2m}\cdot
        \frac{1\pm \varepsilon/c}{d} \subseteq
        \frac{1\pm \varepsilon}{2m}.
    \end{equation*}

    Therefore, for the case $d_{\mathsf{o}} \geq c'/\varepsilon$, the probability
    of sampling the edge $(u,v)$ is $(1\pm \varepsilon)/m$. Similarly, the
    statement holds for the case $d_{\mathsf{o}} < c'/\varepsilon$.
\end{proof}

\begin{lemma}
    The expected running time of the algorithm is
    $O(\varepsilon^{-4}n\log^2n + \varepsilon^{-2}n\log^2n\cdot\log\log n)$.
\end{lemma}
\begin{proof}
    The first step of the algorithm takes $O(n\log^2n)$ expected time.
    The second step of the algorithm takes
    $O(n(\varepsilon^{-2}+\log\log n)\log n)$ expected time. In the third
    step of the algorithm, each iteration requires querying the
    data-structure of \hyperref[cor:disks-estimate]{Corollary~\ref*{cor:disks-estimate}} $O(1)$ times in
    expectation.  Therefore, the third step takes
    $O(|F|)\cdot O((\varepsilon^{-2}+\log\log n)\log n)= O(\varepsilon^{-4}n\log^2n
    + \varepsilon^{-2}n\log^2n\cdot\log\log n)$ expected time.~
\end{proof}

The above implies the following result.

\begin{lemma}
    \label{lemma:random-sampler}%
    Let $\mathcal{D}$ be a collection of $n$ disks in the plane, and let
    $\mathsf{G}_{\cap \mathcal{D}}$ be the corresponding geometric intersection graph
    with $m$ edges. Let $F$ be a random sample of
    $O( \varepsilon^{-2} n\log n)$ edges from $\mathsf{G}_{\cap \mathcal{D}}$, with
    repetition, where the probability of any specific edge to be
    picked is $(1\pm \varepsilon)/m$.  The edges are all chosen independently
    into $F$.  Then, the algorithm described above, for
    computing $F$, runs in
    $O(\varepsilon^{-4}n\log^2n + \varepsilon^{-2}n\log^2n\cdot\log\log n)$ expected
    time.
\end{lemma}

\subsection{The result}

\begin{theorem}
    {\label{theo:second}}%
    Let $\mathcal{D}$ be a collection of $n$ disks in the plane.  A
    $(1+\varepsilon)$-approximation to the densest subgraph of $\mathcal{D}$ can be
    computed in
    $O(\varepsilon^{-4}n\log^2n + \varepsilon^{-2}n\log^2n\cdot\log\log n)$ expected
    time.  The correctness of the algorithm holds with high
    probability.
\end{theorem}
\begin{proof}
    The case of intersection graph having $O(\varepsilon^{-2}n\log n)$ edges
    can be handled directly by computing the whole intersection graph
    in $O(\varepsilon^{-2}n\log n)$ expected time (using \hyperref[lemma:sweepline]{Lemma~\ref*{lemma:sweepline}}).

    To handle the other case, we use \hyperref[lemma:random-sampler]{Lemma~\ref*{lemma:random-sampler}} to
    generate a graph $H=(\mathsf{V}, F)$ in
    $O(\varepsilon^{-4}n\log^2n + \varepsilon^{-2}n\log^2n\cdot\log\log n)$ expected
    time.  Since the intersection graph has $\Omega(\varepsilon^{-2}n\log n)$
    edges, using \hyperref[lemma:sampling]{Lemma~\ref*{lemma:sampling}}, it suffices to compute a
    $(1-\varepsilon/6)$ approximate densest subgraph of $H$, which can be
    computed by the algorithm of \cite{bgm-epgam-14} in
    $O(\varepsilon^{-4}n\log^2n)$ expected time.
\end{proof}

\section{Conclusions}

We presented two near-linear time approximation algorithms to compute
the densest subgraph on (implicit) geometric intersection graph of
disks. We conclude with a few open problems.  It seems that the
running time of the $(2+\varepsilon)$-approximation algorithm can be improved
to $O_\varepsilon(n\log n)$: the deepest point in the arrangement of the
disks and the densest subgraph are mutually bounded by a constant
factor, and a $(1+\vartheta)$-approximation of the deepest point in the
arrangement of the disks can be computed in $O(\vartheta^{-2}n\log n)$
time~\cite{ah-adrp-08}.

Are there implicit geometric intersection graphs, such as unit-disk
graphs or say, interval graphs, for which the \emph{exact} densest
subgraph can be computed in \emph{sub-quadratic time} (in terms of
$n$)?  Finally, maintaining the approximate densest subgraph in
sub-linear time (again in terms of $n$) under insertions and deletions
of objects, looks to be a challenging problem (in prior work on
general graphs an edge gets deleted or inserted, but in an
intersection graph a vertex gets deleted or inserted).

\printbibliography

\appendix

\section{Chernoff's inequality}

The following are standard forms of Chernoff's inequality, see
\cite{mr-ra-95}.

\begin{lemma}
    \label{lemma:chernoff:delta:big}%
    Let $X_1, \ldots, X_n$ be $n$ independent Bernoulli trials, where
    \begin{math}
        {\mathbb{P}}\mleft[  X_i = 1  \mright] =p_i,
    \end{math}
    and
    \begin{math}
        {\mathbb{P}}\mleft[  X_i = 0  \mright] =1- p_i,
    \end{math}
    for $i=1,\ldots, n$.  Let $X = \sum_{i=1}^{b} X_i$, and
    $\mu = {\mathbb{E}}\mleft[ \bigl. X \mright] = \sum_i p_i$.  For $\delta > 2e-1$, we have
    ${\mathbb{P}}\mleft[  \bigl. X > (1+\delta)\mu  \mright] < 2^{-\mu (1 + \delta)}$.
\end{lemma}

\begin{theorem}
    {\label{theo:special:h:i:e}}%
    Let $\varepsilon \in (0,1)$ be a parameter.  Let
    $X_1, \ldots, X_n \in [0,1]$ be $n$ independent random variables,
    let $X= \sum_{i=1}^n X_i$, and let $\mu = {\mathbb{E}}\mleft[ X \mright]$. We have that
    \begin{equation*}
        \displaystyle%
        {\mathbb{P}}\mleft[ \bigl.{X} \notin (1\pm \varepsilon) \mu  \mright]%
        \leq%
        2 \exp\mleft({ -\varepsilon^2 \mu / 3}\mright)
    \end{equation*}
\end{theorem}

\end{document}